\documentclass[journal]{IEEEtran}
\usepackage{cite}

\ifCLASSINFOpdf
\else
\fi
\usepackage[cmex10]{amsmath}
\usepackage{algorithm}
\usepackage{algorithmic}

\usepackage{array}

\usepackage{caption}
\usepackage[font=footnotesize]{subfig}
\usepackage{url}

\usepackage{graphicx,listings,psfrag,amsfonts,cases}

\usepackage{pst-sigsys}
\usepackage{pstricks,pst-node,pst-tree}

\usepackage{amsmath,amssymb,amsfonts,graphicx,nicefrac,mathtools,bbm}
\usepackage{amsthm}

\newcommand{\BIT}{\begin{itemize}}
\newcommand{\EIT}{\end{itemize}}
\newcommand{\BNUM}{\begin{enumerate}}
\newcommand{\ENUM}{\end{enumerate}}
 
\def\reals{\mathbb{R}} 
\def\complex{\mathbb{C}} 
\renewcommand{\exp}[1]{\operatorname{exp}\left(#1\right)} 

\def\E{\mathbb{E}} 

\def\Unif{\textnormal{Unif}}

\newtheorem{theorem}{Theorem}
\newtheorem{lemma}{Lemma}
\newtheorem{assumption}{Assumption}
\newtheorem{corollary}{Corollary}

\newcommand{\putFig}[3]{
        \begin{figure}[h!] 
 		\centering
 		\vspace{-2ex}
 		\includegraphics[width=#3]{figs/#1}
		  \caption{#2}
                \label{fig:#1}
                \vspace{-4ex}
        \end{figure} }

\newcommand{\sset}{\mathcal{S}}
\newcommand{\mv}{\Delta\mathbf{v}}

\usepackage{color}

\begin{document}
%
\title{Fast Distribution Grid Line Outage Identification with $\mu$PMU}


\author{Yizheng~Liao,~\IEEEmembership{Student Member,~IEEE,}
        Yang~Weng,~\IEEEmembership{Member,~IEEE,}
        Chin-Woo~Tan,
        Ram~Rajagopal,~\IEEEmembership{Member,~IEEE}
        \vspace{-4ex}
\thanks{Y. Liao, C-W.Tan, R. Rajagopal are with Department of Civil and Environmental
Engineering, Stanford University, Stanford, CA, 94305 USA e-mail: (\{yzliao, tancw,
ramr\}@stanford.edu). Y. Weng is with School of Electrical, Computing, and
Energy Engineering, Arizona State University, Tempe, AZ, 85287 USA e-mail:
yang.weng@asu.edu.}}


\maketitle

\begin{abstract}

%
The growing integration of distributed energy resources (DERs) in urban distribution grids raises various reliability issues due to DER's uncertain and complex behaviors. With a large-scale DER penetration, traditional outage detection methods, which rely on customers making phone calls and smart meters' ``last gasp'' signals, will have limited performance, because the renewable generators can supply powers after line outages and many urban grids are mesh so line outages do not affect power supply. To address these drawbacks, we propose a data-driven outage monitoring approach based on the stochastic time series analysis from micro phasor measurement unit ($\mu$PMU). Specifically, we prove via power flow analysis that the  dependency of time-series voltage measurements exhibits significant statistical changes after line outages. This makes the theory on optimal change-point detection suitable to identify line outages via $\mu$PMUs with fast and accurate sampling. However, existing change point detection methods require post-outage voltage distribution unknown in distribution systems. Therefore, we design a maximum likelihood-based method to directly learn the distribution parameters from $\mu$PMU data. We prove that the estimated parameters-based detection still achieves the optimal performance, making it extremely useful for distribution grid outage identifications. Simulation results show highly accurate outage identification in eight distribution grids with 14 configurations with and without DERs using $\mu$PMU data.
\end{abstract}


%
\IEEEpeerreviewmaketitle
\vspace{-4ex}
\section{Introduction}
The ongoing large-scale integration of distributed energy resources (DERs) makes photovoltaic (PV) power devices (renewable generation), energy storage devices, and electric vehicles ubiquitous. Such a change transitions the urban power grid into sustainable network and reduces the electricity cost and transmission loss \cite{sce2015application}. However, such a change also raises fundamental challenges in system operations. For example, the reverse power flow from residential houses renders the existing protective architecture inadequate. Also, frequent plug-and-charge electric vehicles will worsen power quality, causing transformer overload and voltage flickers \cite{clement2010impact}. Because of these changes on distribution grid, even a small-scale DER integration can disable a distribution grid \cite{dey2010urban}. 
\cite{greentech} shows that the distribution power outages or blackouts caused by newly added uncertainties can cause a loss of thousands to millions of dollars within one-hour, calling for newly designed fault diagnosis approach for distribution grid operation. 

The traditional power outage analysis in distribution grids relies on passive feedback from customer calls. Collected into Customer Information System (CIS), such information is processed in the Outage Management System (OMS) for sending field crews to identify and repair the outage. Due to the human-in-the-loop system design, delay and imprecise outage information causes inefficient detection and slow restoration. Therefore, smart meters with advanced metering infrastructure (AMI) capability were installed recently to send a ``last gasp'' message when there is a loss of power \cite{luan201417}.  \cite{doe2014fault} shows additional fault location, isolation, and service restoration (FLISR) technologies to reduce some negative impact and the interruptions duration.

However, the traditional methods and the recent approaches above will have limited performance during the DER integration. For example, if rooftop solar panels are installed at the customer's premises, the customer can still receive power from renewable generators when there is no power flow in the distribution circuit connecting to the premises. So the (AMI) smart meter at the customer premises cannot report a power outage. Also, the secondary distribution grids are mesh networks in metropolitan areas \cite{rudin2012machine}, making a line outage unnecessarily cause a power outage. However, it is still important to detect, localize, and identify the out-of-service branches for the situation awareness of distribution system operators.

In transmission grids, there have been works that utilize phasor measurement units (PMUs) to identify line outages. For example, phase changes across all buses are compared with potential fault events in \cite{tate2008line}. In \cite{he2010fault}, a transmission grid is formulated as a graphical model and phase angles are used to track the grid topology change. A regularized linear regression is employed to detect power outages in \cite{zhu2012sparse}. The approach in \cite{wei2012change} compares the branch admittance before and after outages. These methods, however, cannot be directly implemented for the distribution grid because $1)$ the DC approximation has poor performance in distribution grids as many systems have non-negligible line loss; $2)$ installing $\mu$PMUs at all buses in distribution grid is expensive and impractical; and $3)$ the topology information is unavailable or inaccurate in distribution grids, because many DERs do not belong to the utilities and their connectivities are unknown to the system operators \cite{liao2015distribution}. 


For resolving the issues above, we proposed a $\mu$PMU-based data-driven approach, where we model $\mu$PMU measurement at each bus as a random variable, so that the distribution grid is modeled as a multi-variate probability distribution. 
We show that a line outage will lead to a change of the statistical dependence between buses, and consequently, a change of the joint distribution. Hence, the outage can be discovered by detecting the change of the multivariate probability distribution. A well-known method to sequentially detect the probability distribution change is change point detection method \cite{tartakovsky2005general}. 
The change point detection methods have been applied to detect outage in transmission grids \cite{chen2016quickest,wei2012change,rovatsos2017statistical}. However, some requirements in the transmission line outage detection do not hold in distribution grids. 
For example, the change-point detection method requires the post-outage probability distribution. Unfortunately, in practice, this probability distribution is hard to obtain in distribution grids because the number of possible post-outage distributions increases exponentially with the growth of the grid size. To overcome this drawback, we propose a maximum likelihood method to learn the unknown post-outage probability distribution parameters from the $\mu$PMU data. 

For localizing the outage location, 
we firstly prove that the voltages of two disconnected buses are conditionally independent, which is subsequently used to find the line outage without knowing the post-outage probability distribution. 
In case that we do not have $\mu$PMU at every bus, we will first use available $\mu$PMU to make fast event detection. Then, relative location will be provided via a network reduction model. Finally, when the slower smart meter data arrives at buses without $\mu$PMUs, we can use the localization approach like the all-$\mu$PMU approach for highly accurate event localization.


The performance of our data-driven outage detection and localization algorithm is verified by simulations on the standard IEEE $8$- and $123$-bus distribution test cases \cite{kersting2001radial} and $6$ European distribution grids \cite{pretticodistribution} with $14$ network configurations. Real smart meter data collected from $110,000$ Pacific Gas and Electric Company (PG\&E) customers are utilized for generating $\mu$PMU data via data interpolation, different outage scenarios, and power flow analysis. 

The rest of the paper is organized as follows: Section~\ref{sec:model} introduces the modeling and the problem of the data-driven power outage detection and localization based on $\mu$PMUs. Section~\ref{sec:outage_detect} uses a proof to justify that the outage can be detected by change point detection method. Section~\ref{sec:outage_identify} presents the outage localization method. A detailed algorithm for outage detection and localization is illustrated as well. Section~\ref{sec:pmu} discusses how to perform outage identification with a limited amount of $\mu$PMUs available in distribution grids. Section~\ref{sec:num} evaluates the performance of the new method and Section~\ref{sec:con} concludes the paper.

\vspace{-2ex}
\section{System Model}\label{sec:model}
In order to formulate the power outage detection problem, we need to describe the distribution grid and its $\mu$PMU data. A distribution grid is defined as a physical network with buses and branches that connect buses. For a distribution grid with $M$ buses, we use $\mathcal{S}=\{1,2,\dots,M\}$ to represent the set of all bus indices. To utilize the time series data provided by $\mu$PMU, the voltage measurement at bus $i$ is modeled as a random variable $V_i$. We use $\mathbf{V}_\mathcal{S} = [V_1,V_2,\dots,V_M]^T$ to denote all voltage random variables in the network, where $T$ denotes the transpose operator. At the discrete time $n$, the noiseless voltage measurement at bus $i$ is $v_i[n] = |v_i[n]|\exp{j\theta_i[n]} \in \complex$, where $|v_i[n]| \in \reals$ denotes the voltage magnitude in per unit and $\theta_i[n] \in \reals$ denotes the voltage phase angle in degrees. All voltages are sinusoidal signals at the same frequency. We use $\mathbf{v}[n] = [v_1[n],v_2[n],\dots,v_M[n]]^T$ to denote a collection of all voltage measurements in a network at time $n$. Thus, $\mathbf{v}[n]$ is the realization of $\mathbf{V}_\sset$ at time $n$. Also, we use $\mathbf{v}^{1:N} = (\mathbf{v}[1],\mathbf{v}[2],\dots,\mathbf{v}[N])$ to denote a collection of all voltage  measurements in the network up to time $N$.



The problem to detect and localize line outages in a distribution grid is defined as follows:
\begin{itemize}
\item Problem: data-driven power outage detection and localization based on $\mu$PMU measurements
\item Given: a sequence of the historical voltage measurements $\mathbf{v}^{1:N}$ up to the current time $N$ from $\mu$PMUs
\item Find: (1) the outage time and (2) the branches that are out-of-service
\end{itemize}

\section{Optimal Distribution Grid Line Outage Detection}
\label{sec:outage_detect}
Voltage measurements usually have an irregular distribution and are hard to be used for our goal of this paper. Therefore, instead of using voltage measurements directly, we use the incremental change of the voltage measurements from $\mu$PMUs to detect outages, which is defined as $\mv[n] = \mathbf{v}[n] - \mathbf{v}[n-1]$. Accordingly, $\Delta\mathbf{v}^{1:N} = (\mv[1],\mv[2],\cdots,\mv[N])$. We use $\Delta V_i$ to represent the voltage change random variable at bus $i$ and $\Delta\mathbf{V}_\mathcal{S}$ to represent the voltage change random variables of the entire system. In the following, we will prove that, the probability distribution of $\Delta\mathbf{V}_\mathcal{S}$ will be different after an outage. In the following context, the operator $\backslash$ denotes the complement operator, i.e. $\mathcal{A}\backslash\mathcal{B} = \{i \in \mathcal{A}, i\notin \mathcal{B}\}$.

\begin{assumption}
\label{ass:indept}
	In distribution grids, 
	\begin{itemize}
		\item the incremental change of the current injection $\Delta I$ at each non-slack bus is independent, i.e., $\Delta I_i \perp \Delta I_k$ for all $i \neq k$,
		\item the incremental changes of the current injection $\Delta I$ and bus voltage $\Delta V$ follow Gaussian distribution with zero means and non-zero variances. 
	\end{itemize}
\end{assumption}

The Assumption \ref{ass:indept} has been adopted in many works, such as \cite{deka2015structure, bolognani2013identification, liao2018urban}. In \cite{liao2018urban}, the authors use real-data to validate both assumptions. With Assumption \ref{ass:indept}, we prove that the pairwise bus voltages are conditionally independent if there is no branch between them.

\begin{theorem}
\label{thm:cond_indept}
If the change of current injection at each bus is approximately independent and no branch connects bus $i$ and bus $j$, the voltage changes at bus $i$ and bus $j$ are conditionally independent, given the voltage changes of all other buses, i.e. $\Delta V_i \perp  \Delta V_j | \{\Delta V_e,e \in \sset \backslash\{i,j\}\}$.
\end{theorem}
\begin{proof}
For bus $i$, the current and voltage relationship can be expressed as $\Delta I_i = \Delta V_iY_{ii} - \sum_{e \in \mathcal{N}(i)}\Delta V_eY_{ie}$ with $Y_{ii} = \sum_{e \in \mathcal{N}(i)}Y_{ie}$, where $Y_{ie}$ denote is the $ie$th element of the admittance matrix $Y$ and the neighbor set $\mathcal{N}(i)$ contains the indices of the neighbors of bus $i$, i.e., $\mathcal{N}(i) = \{e \in \mathcal{S}|Y_{ie} \neq 0\}$. If bus $i$ and bus $k$ are not connected, $j \notin \mathcal{N}(i)$ and $Y_{ij} = 0$.  Given $\Delta V_e = \Delta v_e$ for all $e \in \sset\backslash \{i,j\}$, the equation above becomes to
\begin{align}
	\Delta I_i &= \Delta V_iY_{ii} - \sum_{e \in \mathcal{N}(i)}\Delta v_eY_{ie}, \nonumber \\
	\Delta V_i &= \frac{1}{Y_{ii}} (\Delta I_i + \sum_{e \in \mathcal{N}(i)}\Delta v_eY_{ie}).\label{eq:VI} 
\end{align}
Similarly, $\Delta V_j = (\Delta I_j + \sum_{e \in \mathcal{N}(j)}\Delta v_eY_{je})/Y_{jj}$. With the assumption of the current change independence, i.e., $\Delta I_i \perp \Delta I_k$, $\Delta V_i$ and $\Delta V_k$ are conditionally independent given $\Delta \mathbf{V}_{\sset\backslash \{i,k\}}$.
\end{proof}

\putFig{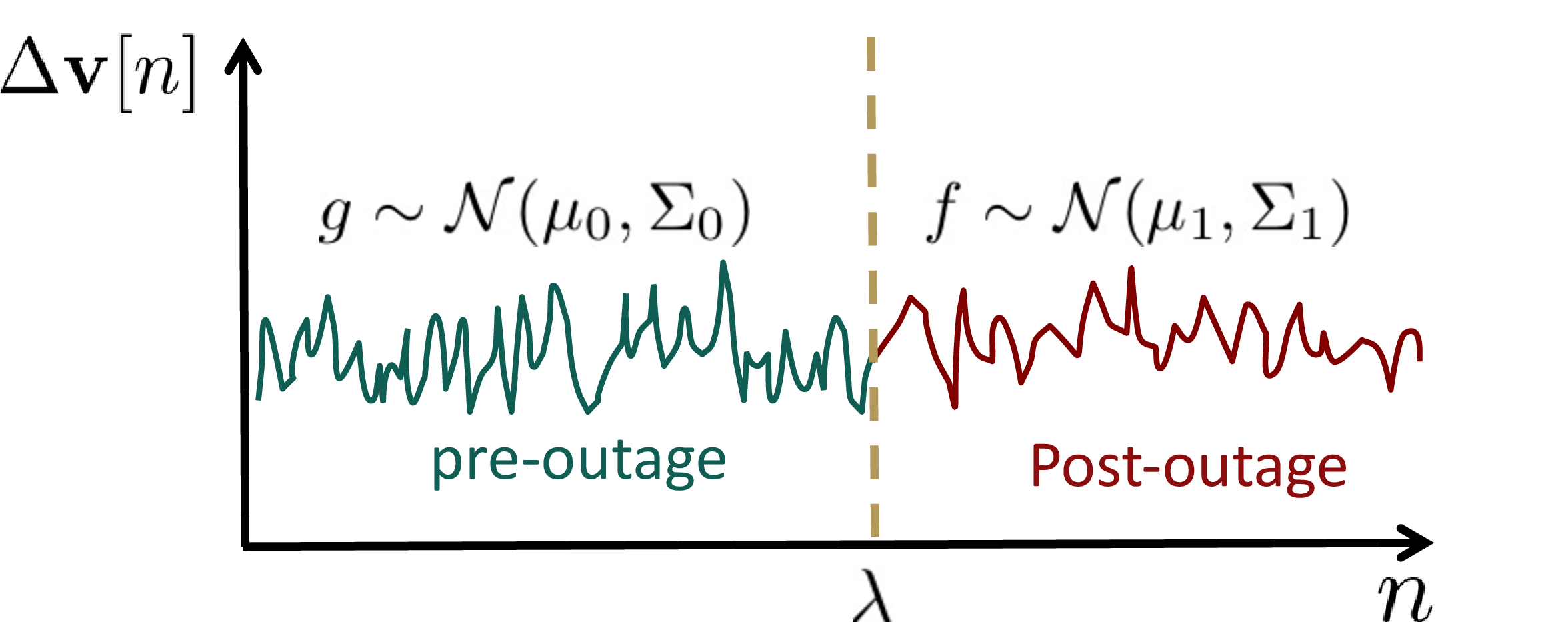}{An example of nodal voltages before and after a line outage. $\lambda$ denotes the outage occurrence time.}{0.8\linewidth}

If a branch is out-of-service, its admittance becomes zero. The voltages at the two ends of this branch becomes conditionally independent. According to Assumption~\ref{ass:indept}, the bus voltage follows a Gaussian distribution. Therefore, $\Delta\mathbf{V}_\mathcal{S}$ follows a multivariate Gaussian distribution. After an outage, some elements of the mean vector and covariance matrix will change. Hence, the probability distribution of $\Delta\mathbf{V}_\mathcal{S}$ is different before and after an outage. Let $\lambda$ denote the time that an outage occurs. We assume that $\Delta\mathbf{V}_\mathcal{S}$ follow a Gaussian distribution $g$ with the mean $\mu_0$ and the covariance $\Sigma_0$ in the pre-outage status (i.e., $N \leq \lambda$) and the other Gaussian distribution $f$ with the mean $\mu_1$ and the covariance $\Sigma_1$ after any outage (i.e., $N > \lambda$). An example is illustrated in Fig.~\ref{fig:outage_plot}. One way to find the outage time $\lambda$ sequentially is to perform the following hypothesis test at each time $N$ \cite{tartakovsky2005general}:
\[
	\mathcal{H}_0: \lambda > N, \quad \mathcal{H}_1: \lambda \leq N.
\]
$\mathcal{H}_0$ represents that the pre-outage status and $\mathcal{H}_1$ represents that the outage has occurred before $N$ (post-outage status). Finding the outage time is known as the change point detection problem. Usually, the line outage occurrence time is unpredictable. Therefore, we assume the power outage time $\lambda$ as a discrete random variable with a probability mass function $\pi(\lambda)$. Now, we can use a Bayesian approach to find $\lambda$. In this paper, we assume $\lambda$ follows a geometric distribution with a parameter $\rho$. The joint distribution of $\lambda$ and $\Delta\mathbf{V}_\mathcal{S}$ can be written as
\[
P(\lambda,\Delta\mathbf{V}_\mathcal{S}) = \pi(\lambda) P(\Delta\mathbf{V}_\mathcal{S}|\lambda).
\]
When $\lambda = k$, all the data obtained before time $k$ follow the distribution $g$ and all the data obtained at and after time $k$ follow the distribution $f$. Therefore, the likelihood probability $P(\Delta\mathbf{V}_\mathcal{S}|\lambda)$ above is expressed as follows:
\[
P(\Delta\mathbf{V}_\mathcal{S} = \mv^{1:N}|\lambda = k) = \prod_{n=1}^{k-1}g(\mv[n])\prod_{n=k}^{N}f(\mv[n]),
\]
for $k = 1,2,\cdots,N+1$. When $\lambda = N+1$, it refers to the outage has not occurred and the data follow the distribution $g$.

Finding the outage time $\lambda$ is equivalent to finding the post-outage posterior probability $P(\mathcal{H}_1|\Delta\mathbf{V}_\mathcal{S}) = P(\lambda \leq N|\Delta\mathbf{V}_\mathcal{S} = \mv^{1:N})$ at each time $N$. If the posterior probability is large enough, we can declare an outage in the system. At each time $N$, 
\begin{align}
	& P(\lambda \leq N|\mv^{1:N}) = \sum_{k=1}^N\frac{P(\lambda = k, \mv^{1:N})}{P(\mv^{1:N})}, \nonumber \\
	=& \frac{1}{P(\mv^{1:N})}\sum_{k=1}^N \pi(\lambda = k)P(\mv^{1:N}|\lambda = k), \nonumber \\
	=& C\sum_{k=1}^N\pi(k)\prod_{n=1}^{k-1}g(\mv[n])\prod_{n=k}^{N}f(\mv[n]), \label{eq:post}
\end{align}
where $C$ is a normalization factor such that $\sum_{k=1}^{N+1}P(\lambda = k|\mv^{1:N}) = 1$. In the normal operation, $f(\mv[n])$ is small and $P(\lambda \leq N|\mv^{1:N})$ is small. Once an outage occurs at time $\lambda = k \leq N$, all data collected at $n \geq \lambda$ follow $f(\mv[n])$ and $P(\lambda \leq N|\mv^{1:N}$ becomes large. Hence, we can set a threshold and declare an outage when the posterior probability surpasses this threshold. This process is visualized in Fig.~\ref{fig:detect_plot}.
\putFig{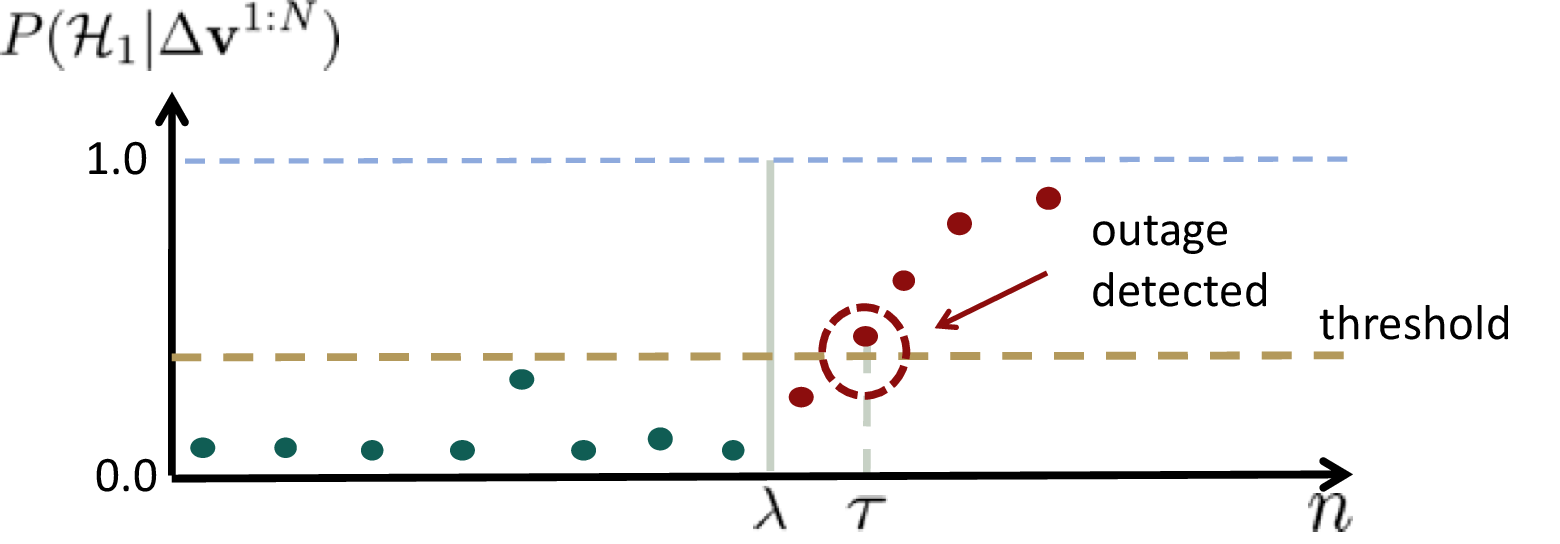}{An example of outage detection based on the posterior probability. $\lambda$ is the outage occurrence time. $\tau$ is the outage detection time. The brown dashed line is the detection threshold.}{0.85\linewidth}

\vspace{-3ex}
\subsection{Optimal Outage Detection}
In the change point detection problem, there are two performance metrics: \textit{probability of false alarm} and \textit{expected detection delay}. The former metric is the probability that a detector falsely declares an outage in the normal operation. If $\tau$ denotes the time of an outage being detected, the probability of false alarm is defined as $P(\tau < \lambda)$. The latter metric describes the average latency that a detector finds the outage after it has occurred. The expected detection delay is defined as $E(\tau - \lambda | \tau \geq \lambda)$. For distribution grid line outage detection via $\mu$PMUs, we want to find the outage time $\lambda$ as quickly as possible with a constraint of the maximum probability of false alarm $\alpha$, i.e.,
\begin{equation}
\label{eq:obj}
\begin{aligned}
& \underset{\tau}{\text{minimize}}
& & E(\tau - \lambda | \tau \geq \lambda) \\
& \text{subject to}
& & P(\tau < \lambda ) \leq \alpha.
\end{aligned}
\end{equation}
By the Shiryaev-Roberts-Pollaks procedure \cite{pollak2009optimality}, we have the following lemma to solve the optimization problem in (\ref{eq:obj}).

\begin{lemma}
\label{thm:detect_rule}
	Given a maximum probability of false alarm $\alpha$, the following detection rule
	\begin{equation}
	\label{eq:detect_rule}
	\tau = \inf\{N \geq 1: P(\lambda \leq N|\mv^{1:N}) \geq 1-\alpha\},
\end{equation}
is asymptotically optimal \cite{tartakovsky2005general}.
\end{lemma} 
With Lemma~\ref{thm:detect_rule}, the threshold (brown dashed line) in Fig.~\ref{fig:detect_plot} is $1-\alpha$. Lemma~\ref{thm:opt_delay} shows the asymptotically optimal expected detection delay.
\begin{lemma}
	\label{thm:opt_delay}
	For a given probablity of false alarm $\alpha$, the detection rule in (\ref{eq:detect_rule}) achieves the asymptotically optimal detection delay
	\begin{equation}
	\label{eq:bound}
	D(\tau) = E(\tau - \lambda|\tau \geq \lambda) = \frac{|\log\alpha|}{-\log(1-\rho)+D_\text{KL}(f\|g)},
	\end{equation}
	as $\alpha \rightarrow 0$, where $D_\text{KL}(f\|g)$ is the Kullback-Leibler distance \cite{tartakovsky2008asymptotically}.
\end{lemma}

In summary, when a new group of observation $\mv[n]$ is available from $\mu$PMUs, we can  compute the posterior probability according to (\ref{eq:post}) and then apply the optimal detection rule in (\ref{eq:detect_rule}) to diagnose the system. Therefore, the proposed algorithm can be implemented for real-time outage detection. As a highlight, the proposed approach does not require the grid topology.

\vspace{-2ex}
\subsection{Line Outage Detection with Unknown Outage Pattern}
Computing the posterior probability in (\ref{eq:post}) requires both distributions $g$ and $f$. The parameters of pre-outage distribution $g$ can be estimated using the historical measurements from $\mu$PMUs. For obtaining the post-outage distribution parameters, we need to know the outage pattern as a prior. One way is trying every possible outage pattern and identify the most similar one. However, this approach is infeasible because the outage patterns can grow exponentially with the grid size \cite{zhao2016efficient}. Also, many DERs in distribution grids are not operated by the utilities. Therefore, their topology information is usually unknown \cite{weng2017distributed}.

In this section, instead of searching the most likely post-outage distribution, we propose a method to learn $f$ from data using the maximum likelihood method. The computational complexity of our approach is insensitive to the number of out-of-service branches. To apply the maximum likelihood method, we need to compute the partial derivative of the posterior probability $P(\mathcal{H}_1|\Delta\mathbf{v}^{1:N})$. Unfortunately, $P(\mathcal{H}_1|\Delta\mathbf{v}^{1:N})$ is not a convex function and we may have multiple estimates. To address this challenge, we will provide an approximation of the posterior probability $P(\mathcal{H}_1|\Delta\mathbf{v}^{1:N})$. Specifically, the log-probability $\log P(\mathcal{H}_1|\Delta\mathbf{v}^{1:N})$ is
\begin{align}
&	\log P(\mathcal{H}_1|\Delta\mathbf{v}^{1:N}) \nonumber \\
=& \log C + \log\left\{\sum_{k=1}^N\pi(k)\prod_{n=1}^{k-1}g(\mv[n])\prod_{n=k}^{N}f(\mv[n]; \boldsymbol{\Theta})\right\}, \label{eq:log_post}
\end{align}
where $\boldsymbol{\Theta} = \{\mu_1, \Sigma_1\}$ represents the unknown parameters of $f$. In (\ref{eq:log_post}), the term within the braces is an expectation of $\prod_{n=1}^{k-1}g(\mv[n])\prod_{n=k}^{N}f(\mv[n]; \boldsymbol{\Theta})$ over the prior distribution $\pi$, $\E_\pi(\prod_{n=1}^{k-1}g(\mv[n])\prod_{n=k}^{N}f(\mv[n]; \boldsymbol{\Theta}))$. Also, the logarithmic function is convex. Therefore, we can apply the Jensen's inequality \cite{cover2012elements} to approximate $\log P(\mathcal{H}_1|\Delta\mathbf{v}^{1:N})$:
\begin{align}
	& \log P(\mathcal{H}_1|\Delta\mathbf{v}^{1:N}) \nonumber \\
	\geq & \log C + \sum_{k=1}^N\pi(k)\left(\sum_{n=1}^{k-1}\log g(\mv[n]) + \sum_{n=k}^{N} \log f(\mv[n]; \boldsymbol{\Theta})\right) \nonumber \\
	= & \widetilde{P}(\mathcal{H}_1|\Delta\mathbf{v}^{1:N}) \label{eq:approx_log_post}.
\end{align}
Since $\widetilde{P}(\mathcal{H}_1|\Delta\mathbf{v}^{1:N})$ is convex, by setting $\partial\widetilde{P}/\partial\mu_1 = 0$ and $\partial\widetilde{P}/\partial\Sigma_1 = 0$, the parameter $\boldsymbol{\Theta}$ is estimated as
\begin{align}
	\widehat{\mu}_1 & = \frac{\sum_{k=1}^N \pi(k)\sum_{n=k}^N \mv[n]}{\sum_{k=1}^N \pi(k) (N-k+1)}, \label{eq:mu_est} \\
	\widehat{\Sigma}_1 & = \frac{\sum_{k=1}^N \pi(k)\sum_{n=k}^N (\mv[n] - \widehat{\mu}_1)(\mv[n] - \widehat{\mu}_1)^T}{\sum_{k=1}^N \pi(k) (N-k+1)} \label{eq:sigma_est}.
\end{align}
The details of (\ref{eq:mu_est}) and (\ref{eq:sigma_est}) are given in Appendix~\ref{sec:parm_est}. With the estimates of $\mu_1$ and $\Sigma_1$, we can compute the posterior probability in (\ref{eq:post}) and apply the optimal detection rule in (\ref{eq:detect_rule}).

\vspace{-2ex}
\section{Line Outage Identification}
\label{sec:outage_identify}
Identifying which line has an outages is important in the urban distribution grid operation. In metropolitan areas, many branches are underground and not well documented. Therefore, an efficient and accurate outage localization approach can reduce the power interruption time significantly. In the following part, we will propose a real-time outage localization method based on the voltage measurements from $\mu$PMUs.

\begin{lemma}
\label{thm:cond_cov}
    Assuming random vectors $\mathbf{X}$, $\mathbf{Y}$, and $\mathbf{Z}$ follow Gaussian distributions, given $\mathbf{Z} = \mathbf{z}$, if $\mathbf{X}$ and $\mathbf{Y}$ are conditionally independent, their conditional covariance is zero \cite{hastie2015statistical}. 
\end{lemma}

Because of Theorem~\ref{thm:cond_indept}, the voltage changes at the two ends of the out-of-service branches are conditionally independent after an outage. Due to Lemma~\ref{thm:cond_cov}, we can compute the conditional covariance matrix of every possible pair of buses in the network and check if the off-diagonal term changes to zero. When the off-diagonal term changes to zero, we can identify the out-of-service branches. 

Usually, the conditional covariance can be estimated based on the voltage measurements. However, a large set of post-outage data is required to have an accurate estimation, and the delay of localization is long. To enable real-time outage localization based on $\mu$PMUs, we use the covariance matrix $\Sigma$ to compute the conditional covariance alternatively. This approach allows us to localize the outage even if we do not know the distribution grid topology. In the case that the post-outage probability distribution $f$ is unknown, we can use $\widehat{\Sigma}_1$ in (\ref{eq:sigma_est}) to compute the conditional covariance. For bus $i$ and bus $j$, suppose $\mathcal{I} = \{i,j\}$ and $\mathcal{J} = \mathcal{S}\backslash\{i,j\}$, the covariance of the joint Gaussian distribution can be decomposed as
\[
\Sigma = \begin{bmatrix}
	\Sigma_{\mathcal{I}\mathcal{I}} & \Sigma_{\mathcal{I}\mathcal{J}} \\ 
	\Sigma^T_{\mathcal{I}\mathcal{J}} & \Sigma_{\mathcal{J}\mathcal{J}}\end{bmatrix}.
\]
The conditional covariance matrix can be computed by the Schur complement \cite{boyd2004convex}, i.e.,
\begin{equation}
	\label{eq:cond_cov}
	\Sigma_{\mathcal{I}|\mathcal{J}} = \Sigma_{\mathcal{I}\mathcal{I}} - \Sigma_{\mathcal{I}\mathcal{J}}\Sigma_{\mathcal{J}\mathcal{J}}^{-1}\Sigma^T_{\mathcal{I}\mathcal{J}}.
\end{equation}
If the voltages at bus $i$ and bus $j$ are conditionally independent, the off-diagonal term of $\Sigma_{\mathcal{I}|\mathcal{J}}$ is zero, i.e. $\Sigma_{\mathcal{I}|\mathcal{J}}(1,2) = \Sigma_{\mathcal{I}|\mathcal{J}}(2,1) = 0$. Therefore, we can compare the conditional covariance of every bus pairs before and after an outage. If the conditional covariance changes to zero after an outage, we localize one line outage. This computation can be repeated when $\widehat{\Sigma}_1$ is updated based on the latest available measurements. In Section~\ref{sec:num}, we illustrate the similar performances using the true post-outage covariance matrix $\Sigma_1$ and the estimated covariance matrix $\widehat{\Sigma}_1$.

Let $\widehat{\mathcal{E}}_\text{outage}$ denote the estimated out-of-service branch set. We summarize the proposed outage identification algorithm in Algorithm~\ref{alg:outage}.
\vspace{-2.5mm}
\begin{algorithm}[h!]
\caption{Distribution Grid Line Outage Identification}
\label{alg:outage}
\begin{algorithmic}[1]
\STATE At each time $N$:
\IF {parameters of post-outage distribution $f$ are unknown}
    \STATE estimate $\widehat{\mu}_1$ and $\widehat{\Sigma}_1$ using (\ref{eq:mu_est}) and (\ref{eq:sigma_est})
\ENDIF
\STATE Compute $P(\mathcal{H}_1|\mv_\sset^{1:N})$ by (\ref{eq:post}).
\IF {$P(\mathcal{H}_1|\mv_\sset^{1:N}) \geq 1-\alpha$}
    \STATE Report an outage event and $\tau = N$
    \STATE Compute $\Sigma_{\mathcal{I}|\mathcal{J}}$ by (\ref{eq:cond_cov}) using $\Sigma_1$ for every pair of buses 
    \IF {$\Sigma_{\mathcal{I}|\mathcal{J}} = 0$ for $\mathcal{I} = \{i,j\}$}
        \STATE Report the branch between bus $i$ and bus $j$ is out-of-service
    \ENDIF
\ENDIF
\end{algorithmic}
\end{algorithm}
\vspace{-4mm}

    
    


\vspace{-2ex}
\section{Line Outage Identification with Limited $\mu$PMU Deployment}
\label{sec:pmu}
Due to recent investment on AMI infrastructure, smart meters have been widely installed. However, $\mu$PMUs are only installed at selected buses because of its high cost. The line outage detection and localization algorithm proposed in the previous sections assumes that every bus has a $\mu$PMU. In this section, we will extend the algorithm to two particular cases:
\begin{itemize}
	\item Combining $\mu$PMUs and smart meters data to identify line outages, and
	\item Only using $\mu$PMUs data to identify line outages.
\end{itemize}

\subsection{Using $\mu$PMUs and smart meters data to identify line outages}

In distribution grids, the sampling frequencies of $\mu$PMU and smart meters are usually different. For $\mu$PMU, the sampling frequency is usually $120$Hz. For smart meters, the data collection period depends on the device but is usually larger than one minute per sample. To utilize the $\mu$PMU and smart meters together, we can perform the sequential test when both $\mu$PMU and smart meter data are updated. Specifically, given $\widetilde{\mv}_\sset[n] = \{\mv_{\mathcal{A}}[n],|\Delta\mathbf{v}_{\mathcal{B}}[n]|\}$, we compute $P(\mathcal{H}_1|\widetilde{\mv}_\sset[n])$ and apply the decision rule in (\ref{eq:detect_rule}) to diagnose distribution grids, where $\mathcal{A}$ denotes the set of buses that are installed with $\mu$PMU and $\mathcal{B}$ denotes the set of buses that only have smart meters. If we assume the phase angles of the smart meter buses are zero, Lemma~\ref{thm:cond_indept} still holds on these buses \cite{liao2018urban}. Therefore, we can apply the same method in (\ref{eq:cond_cov}) to identify out-of-service branches.
\vspace{-2ex}
\subsection{Only using $\mu$PMUs data to identify line outages}
\label{sec:pmu_localize}
The sampling frequency of the smart meter is usually much lower than that of $\mu$PMU. If we wait until all smart meters data are available, the detection delay will increase significantly. Even if we only need one sample ($\mv_\sset[n]$) to report outage events, the detection delay is as long as several minutes. In order to quickly detect service-critical outages, we extend the proposed approach to only utilize $\mu$PMU measurements of a subset of bus.

\begin{corollary}
\label{thm:pmu}
	When a distribution grid has a limited amount of $\mu$PMU installed, we can still detect line outages using the method in (\ref{eq:detect_rule}).
\end{corollary}
\begin{IEEEproof}
Let $\mathcal{A}$ denote the set of buses that are installed with $\mu$PMU, $\mathcal{B}$ denote the set of buses that do not have $\mu$PMU and $\mathcal{A} \cup \mathcal{B} = \sset$. We can partition a distribution grid as follows \cite{weng2017geometric}:
\[
\begin{bmatrix}
	\mathbf{I}_{\mathcal{A}} \\ 
	\mathbf{I}_{\mathcal{B}}
\end{bmatrix}
= \begin{bmatrix}
	\mathbf{Y}_{\mathcal{A},\mathcal{A}} & \mathbf{Y}_{\mathcal{A},\mathcal{B}} \\
	\mathbf{Y}_{\mathcal{B},\mathcal{A}} & \mathbf{Y}_{\mathcal{B},\mathcal{B}}
\end{bmatrix}
\begin{bmatrix}
	\mathbf{V}_{\mathcal{A}} \\ 
	\mathbf{V}_{\mathcal{B}}	
\end{bmatrix}.
\]
Then, the distribution grid can be reduced to $\mathbf{I}_{\mathcal{A}} = \widetilde{\mathbf{Y}}\mathbf{V}_{\mathcal{A}}$, where $\widetilde{\mathbf{Y}} = (\mathbf{Y}_{\mathcal{A},\mathcal{A}} - \mathbf{Y}_{\mathcal{A},\mathcal{B}}\mathbf{Y}^{-1}_{\mathcal{B},\mathcal{B}}\mathbf{Y}^T_{\mathcal{A},\mathcal{B}})$. If an outage event occurs, one of the matrices $Y_{\mathcal{A},\mathcal{A}}$, $Y_{\mathcal{A},\mathcal{B}}$, and $Y_{\mathcal{B},\mathcal{B}}$ will be changed. Therefore, $\widetilde{\mathbf{Y}}$ will be different and we can still apply the optimal change point detection method presented in the previous section to detect outages.
\end{IEEEproof}
Similar to the previous applications, we want to localize outage after detection. Unfortunately, the conditional independence property cannot be applied here to identify the out-of-service branch. The reason is that the out-of-service branch may not change the element in $\widetilde{\mathbf{Y}}$ to zero. Therefore, we cannot apply the method in Section~\ref{sec:outage_identify} to localize outage. Alternatively, we propose a two-stage approach to narrow down the potential outage area and help the crew members and engineers to reduce searching time.

Although the conditional covariance between two buses will not become zero after an outage event, its value will still change significantly due to the change of admittance. This observation is inspired by the real data simulation from PG\&E. For example, Fig.~\ref{fig:8bus_loop} illustrates a 8-bus system with loops. For this system, we only install $\mu$PMUs on bus $2$, $4$, $5$, and $8$. Fig.~\ref{fig:pmu_outage_localization} shows the change of conditional correlation before and after branch $2-6$ is out-of-service. We can observe that the conditional correlation between bus $2$ and bus $5$ has the most significant change. Therefore, as the first stage, we can still compute the pairwise conditional covariance and find the ones that have significant change as the line outage candidates.

\putFig{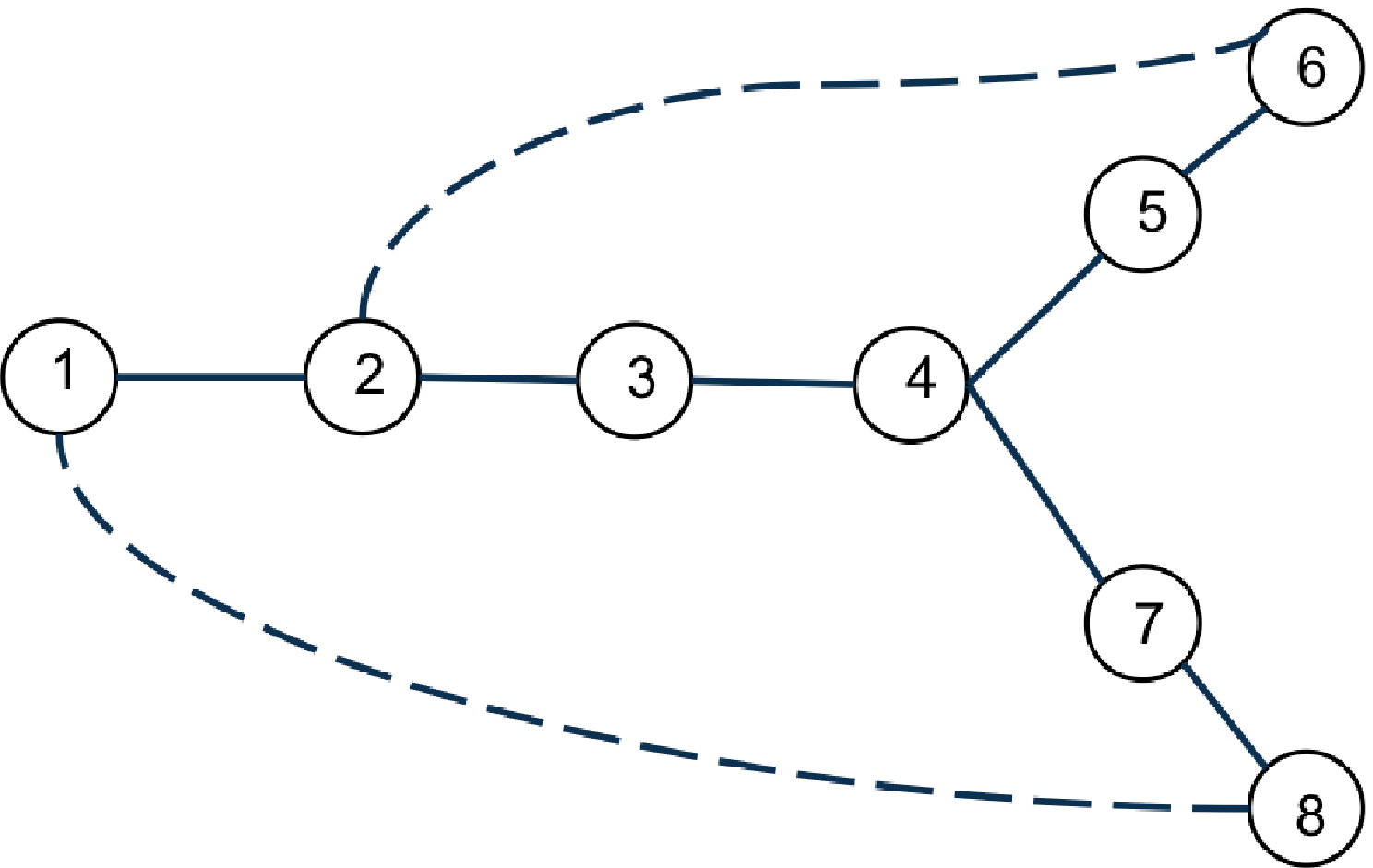}{An 8-bus system. A node represents a bus and a line represents a branch. The dashed lines are additional branches with the same admittance as the branch connected bus $7$ and bus $8$.}{0.5\linewidth}

\begin{figure}[htbp]
    \centering
    \subfloat[Pre-outage\label{fig:pre-outage-pmu}]{
    \includegraphics[width=0.47\linewidth]{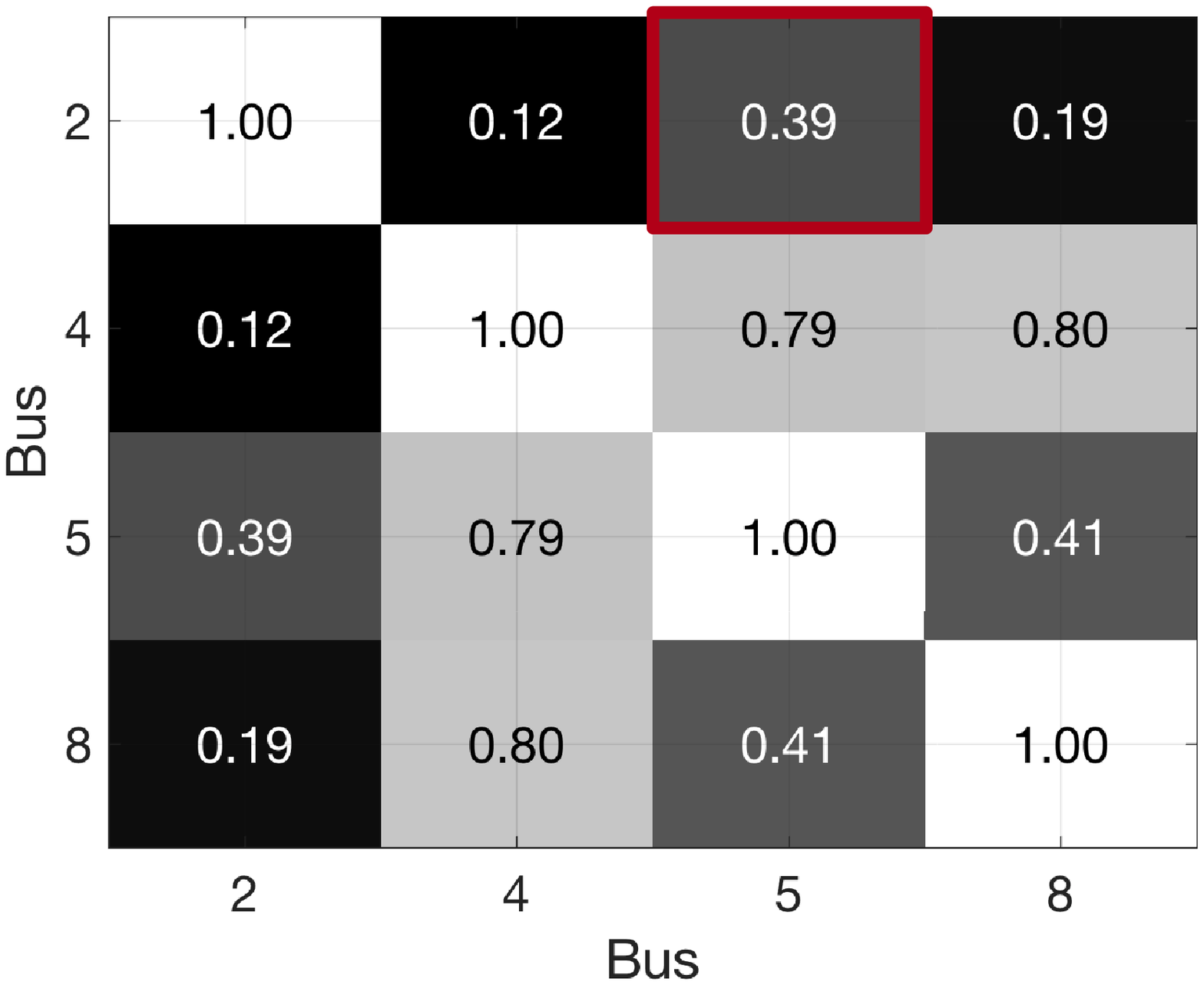}
    }
    \hfill
    \subfloat[Post-outage\label{fig:post-outage-pmu}]{
    \includegraphics[width=0.47\linewidth]{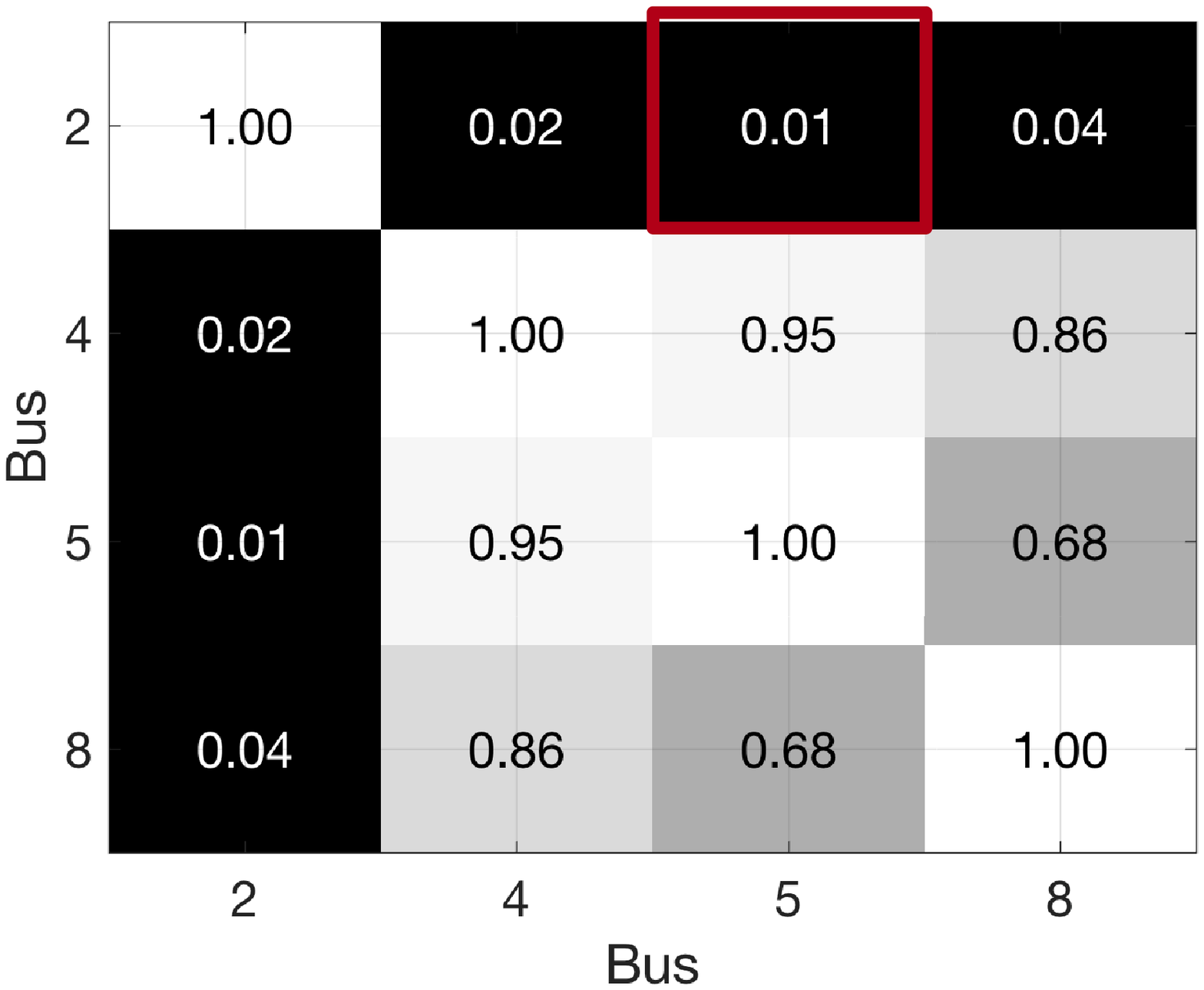}
    }
    \caption{Absolute conditional correlation before and after an outage (Branch 2-6). Only bus $2$, $4$, $5$, and $8$ have $\mu$PMUs.}
    \label{fig:pmu_outage_localization}
    \vspace{-3ex}
\end{figure}


For the newly built distribution grids, the distribution line configurations (e.g., conductor geometric mean radius, resistivity of each in Ohm-meter) and line lengths are available and accurate. By utilizing (\ref{eq:VI}), we can estimate the line admittance before and after outages over the line outage candidates, which are found from the previous step.  Then, we can apply modified Carson's equations to identify possible outage branch length and narrow down the searching area. Computing (\ref{eq:VI}) requires to know all bus connections. We can either use the utility records or the existing distribution grid topology estimation methods \cite{liao2018urban,weng2017distributed,deka2015structure}.

\vspace{-2ex} 
\section{Simulation and Results}\label{sec:num}
The simulations are implemented on the IEEE PES distribution networks for IEEE $8$-bus and $123$-bus networks \cite{kersting2001radial} and six European distribution grids \cite{pretticodistribution}. To validate the performance of the proposed approach on loopy networks, we add several branches to create loops in all systems. The loopy $8$-bus system is shown in Fig.~\ref{fig:8bus_loop}. For $123$-bus system, we add a branch between bus 77 and bus 120 and the other branch between bus 50 and bus 56. The admittance are the same as the branch between bus 122 and bus 123. For European systems, the loopy modifications are detailed in \cite{liao2018urban}. In each network, bus $1$ is selected as the slack bus. The historical data have been preprocessed by the MATLAB Power System Simulation Package (MATPOWER) \cite{Zimmerman10}. 

To simulate the power system behavior in a more realistic pattern, we build $\mu$PMU measurements via interpolation based on real power profile of distribution grids from Pacific Gas and Electric Company (PG\&E) in the subsequent simulation. This profile contains anonymized and secure hourly smart meter readings over $110,000$ PG\&E residential customers for one year spanning from $2011$ to $2012$. The reactive power $q_i[n]$ at bus $i$ and time $t$ is computed according to a randomly generated power factor $pf_i[n]$, which follows a uniform distribution, e.g. $pf_i[n] \sim \Unif(0.8,1)$. The results presented in the this section are based on the complex voltage from $\mu$PMU. 
To obtain measurements form voltage phasors at time $n$, i.e. $v_i[n]$, we run a power flow to generate the states of the power system. To obtain time-series data, we run the power flow to generate fast scale $\mu$PMU data over a year. 


In this simulation, we considered three common outage scenarios:
\begin{itemize}
    \item A loopy network. In this system, after an outage, most buses will not have zero voltages because they can receive powers from multiple branches. This outage scenario usually happens in the urban areas.
    \item A radial network with DERs. In this case, some buses will be disconnected from the main grid. However, if they are connected with DERs, their voltages will not be zero. This outage case is a typical scenario in the residential areas. 
    \item A radial network without DERs. In this case, when a line outage occurs, some buses will be disconnected from the main grid and have zero voltage magnitudes. Because the bus voltages have no variation after outages, our method can quickly detect and localize this type of outages.
\end{itemize}
\vspace{-2ex}
\subsection{Outage Detection in Loopy Systems}
Fig.~\ref{fig:8bus_loop_det} illustrates the complimentary posterior probability $1-P(\mathcal{H}_1|\mv^{1:N})$ for detecting two line outages in loopy 8-bus system (Fig.~\ref{fig:8bus_loop}) based on $\mu$PMU data. In this test, Branches $3$-$4$ and $2$-$6$ have outages. The false alarm rate $\alpha$ is $10^{-6}$. For the complimentary posterior probability, the threshold is $\alpha = 10^{-6}$.
To have a better understanding of how our proposed outage detection algorithm works, we assign a uninformative parameter for the prior distribution, i.e. $\rho = 10^{-4}$. The outage time is $\lambda = 21$. When the parameters of post-outage distribution are known, the complimentary posterior probability immediately drops blow the threshold at $N=21$. When the parameters are unknown, one more time step is required to achieve detectable probability. Given that the $\mu$PMU data are collected at a fast time scale, this delay is  insignificant in field applications.

\putFig{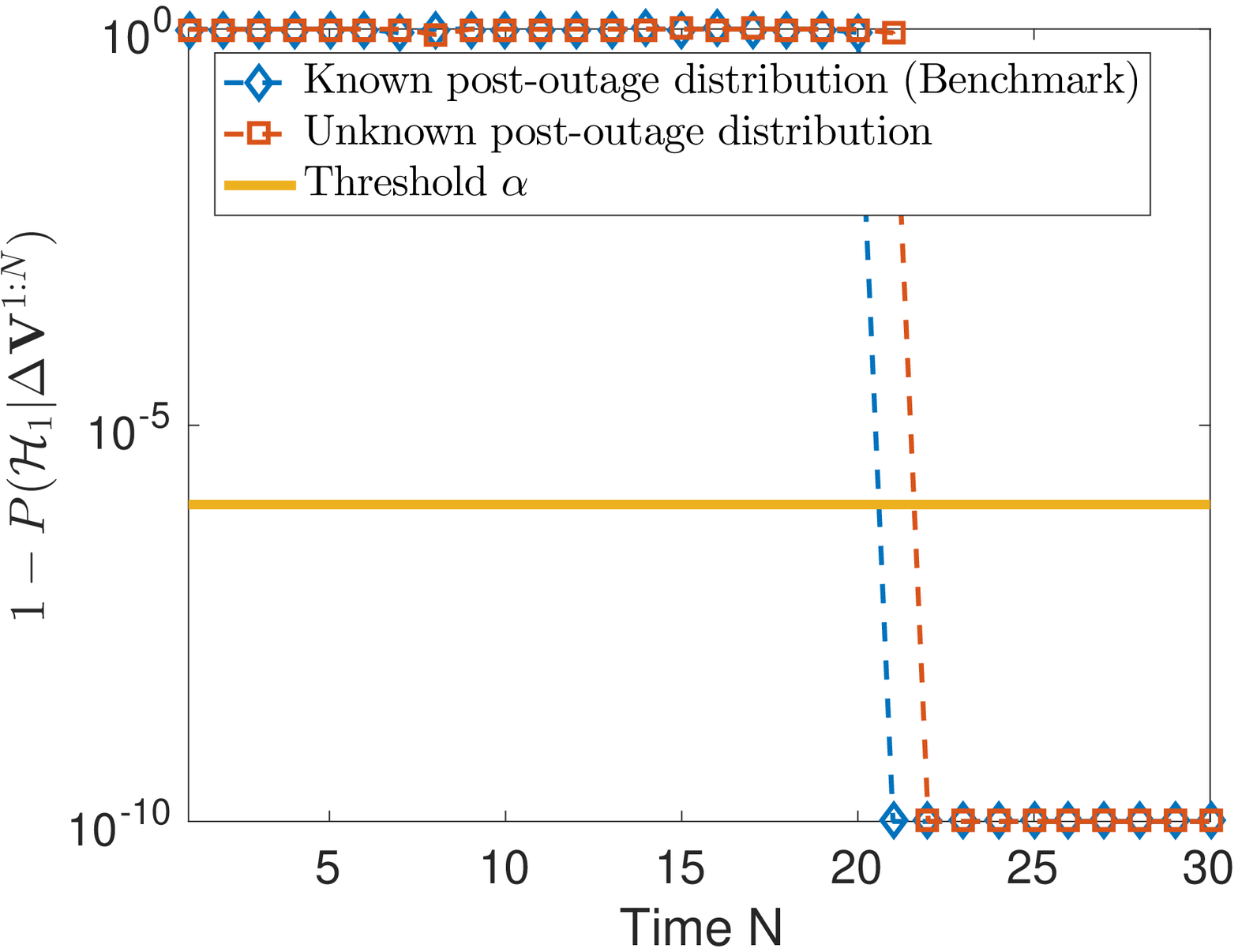}{Complimentary posterior probability for outage detection. The branches $3$-$4$ and $2$-$6$ have outage. $\alpha = 10^{-6},\rho = 10^{-4}$.}{0.8\linewidth}

In Fig.~\ref{fig:123busloop_delay}, the expected delay divided by $|\log\alpha|$ is plotted as a function of $|\log\alpha|$ for two cases: $f$ is known and $f$ is unknown. We also show the limiting value of the normalized asymptotically optimal detection delay as predicted by Lemma~\ref{thm:opt_delay}. All plots are generated by Monte Carlo simulation over $1,000$ replications. In this simulation, the prior distribution of outage time $\lambda$ has a geometric probability distribution with parameter $\rho = 0.04$. The start time of test is randomly selected within one year. In Fig.~\ref{fig:123busloop_delay}, our approach, which learns the parameters of the post-outage distribution from the voltage measurements, has identical performances as the optimal method that has known $f$. Also, our approach can achieve the optimal expected detection delay asymptotically. As shown in Fig.~\ref{fig:123busloop_delay}, when the false alarm rate $\alpha$ is small, our approach can report the outage immediately (i.e. detection delay is less than one time slot), which can significantly reduce the impacts of power outages.
 
\putFig{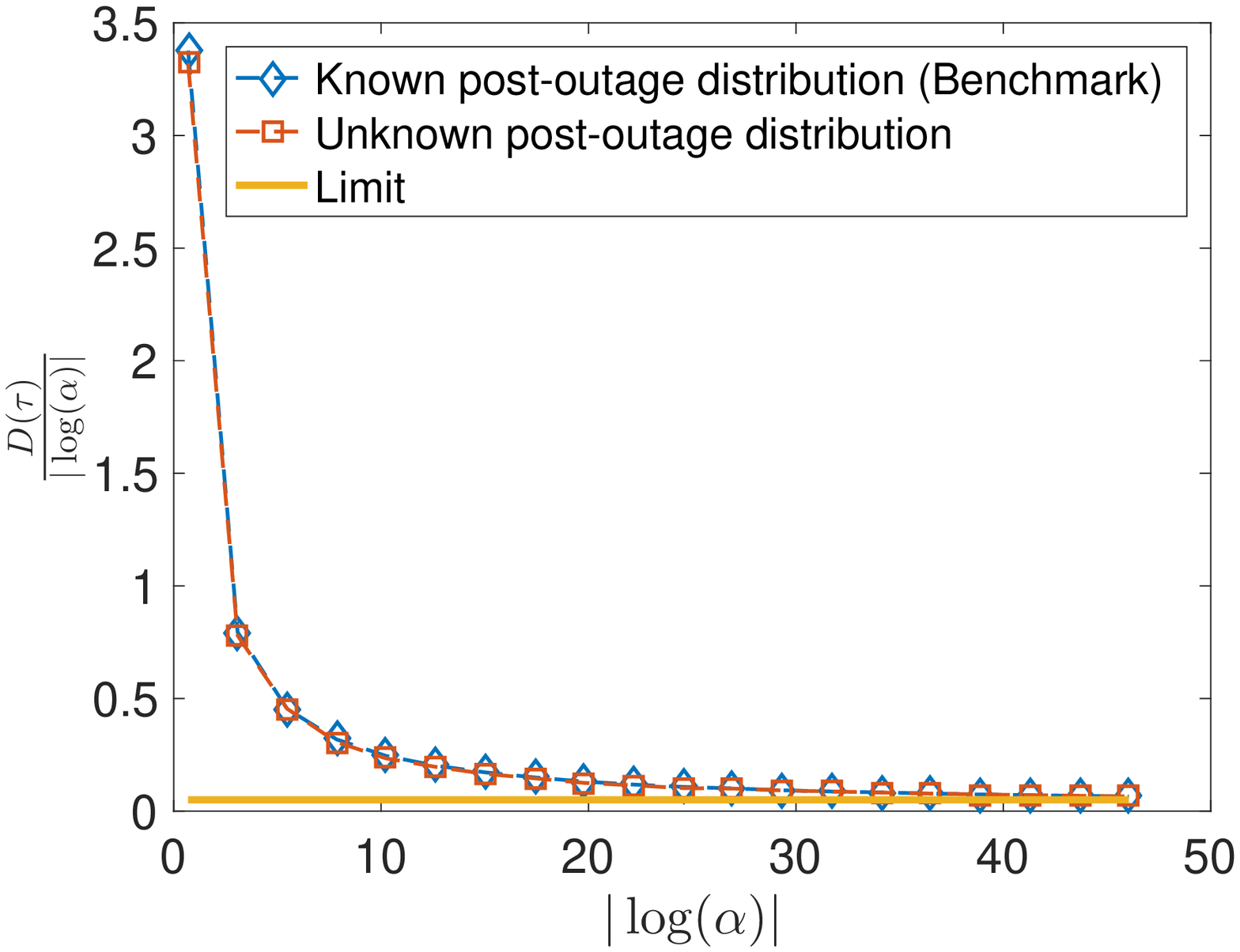}{Plots of the slope $\frac{1}{|\log\alpha|}E[\tau-\lambda|\tau\geq \lambda]$ against $|\log\alpha|$ for outage detection for loopy $123$-bus system. False alarm rate $\alpha$ ranges in $[0.5,10^{-20}]$. Branch $73$-$74$ has an outage.}{0.8\linewidth}

\subsection{Outage Detection in Systems with DERs}
In a radial distribution grid, a line outage will lead to several isolated islands. However, with the integration of DERs, such as solar panels and batteries, some buses can still receive powers. In loopy systems, the continuous power supply from DERs also makes the outage detection difficult. In this section, we simulate the line outage in IEEE 8-bus and 123-bus systems and six European medium- and low-voltage distribution systems based on $\mu$PMU data \cite{pretticodistribution, liao2018urban}. Similar to the previous section, we randomly selected the start time within one year. Also, we select a few buses in the distribution grid to have solar power generator with a battery as the storage. Thus, there is a power supply during the entire day. If the battery is unavailable, the outage can be directly detected when the nodal voltages are zero. For the solar panel, we use the hourly power generation profile computed by PVWatts Calculator, an online application developed by the National Renewable Energy Laboratory (NREL) \cite{dobos2014pvwatts}. The solar power generation profile are computed based on the weather history in North California and the physical parameters of ten 5kW solar panels. The power factor is fixed as $0.90$ lagging, which satisfies the regulation of many U.S. utilities \cite{ellis2012review} and IEEE standard \cite{ieee2014guide}.

\begin{table}[h]
\caption{Average Detection Delay (Time Step) of Line Outage Detection in Distribution Grids with DERs. $\alpha = 10^{-5}$.}
\centering
	\begin{tabular}{|c||c|c|c|c|c|}
	\hline
	System & Total & Total & $\Delta \mathbf{V}_\sset$ & $|\Delta \mathbf{V}_\sset|$ \\
	& Branches & PV & (1 min) & (15 min) \\
	\hline
	8-bus & 7 & 8 & 0.12 & 0.12 \\
	\hline
	8-bus, 2 loops & 9 & 8 & 0.13 & 0.15 \\
	\hline
	123-bus & 122 & 12 & 4.62 & 5.77 \\
	\hline
	123-bus, 2 loops & 124 & 12 & 4.53 & 5.89 \\
	\hline
	\textit{LV\_suburban} & 114 & 10 & 3.81 & 6.01 \\
	\hline
	\textit{LV\_suburban} & 114 & 20 & 3.99 & 6.01 \\
	\hline
	\textit{LV\_suburban} & 114 & 33 & 4.23 & 6.12 \\
	\hline
	\textit{MV\_urban} & 34 & 7 & 1.11 & 2.02  \\
	\hline
	\textit{MV\_urban} & 35 & 7 & 1.11 & 1.29 \\
	switch 34-35, 1 loop &&&& \\ 
	\hline
	\textit{MV\_urban} & 37 & 7  & 1.12 & 1.29 \\
	3 switches, 3 loops &&&& \\ 
	\hline
	\textit{MV\_two\_stations} & 46 & 10 & 0.92 & 1.33 \\
	\hline
	\textit{MV\_two\_stations} & 48 & 10 & 0.87 & 1.35 \\
	2 switches, 2 loops &&&& \\
	\hline
	\textit{MV\_rural} & 116 & 20 & 1.13 & 2.44 \\
	\hline
	\textit{MV\_rural} & 119 & 20 & 1.98 & 3.01 \\
	3 switches, 3 loops &&&& \\
	\hline
	\textit{Urban} & 3237 & 300 & 12.89 & 30.23 \\
	\hline
	\textit{LV\_large}, 465 loops & 4030 & 300 & 34.29 & 194.09 \\
	\hline
	\end{tabular}
\label{tab:detection_DER}
\end{table}

Table~\ref{tab:detection_DER} summarizes the average detection delay in eight distribution grids with $14$ configurations. In each network, we evaluate the detection performance using $\mu$PMU ($\Delta \mathbf{V}_\sset$) and smart meters ($|\Delta \mathbf{V}_\sset|$). The sampling frequency of $\mu$PMUs is $1$ minute to demonstrate the a relative faster metering speed when compared to $15$ minutes for smart meters. However, the result can easily be generalized for faster $\mu$PMU speed. For small networks, the smart meters require one to three more observations than the $\mu$PMU. The number of additional observations is small. But considering the sampling frequencies are not the same, the $\mu$PMU reports outage events much faster than the smart meters. When the distribution grid is large (e.g., \textit{Urban} and \textit{LV\_large} systems), the $\mu$PMU requires much less samples to detect an outage event due to the significant change of phase angles.

\subsection{Line Outage Localization}
The changes of voltage measurements are small in many cases. Thus, the conditional covariance is hard to visualize. Alternatively, we show the absolute conditional correlation in this section for $\mu$PMU data. For bus $i$ and bus $j$, their conditional covariance matrix is denoted as $\Sigma_{\mathcal{I}|\mathcal{J}}$, where $\mathcal{I} = \{i,j\}$ and $\mathcal{J} = \mathcal{S}\backslash\{i,j\}$. The conditional correlation between bus $i$ and bus $j$ is defined as
\[
\rho_{i,j} = \frac{\Sigma_{\mathcal{I}|\mathcal{J}}(1,2)}{\sqrt{\Sigma_{\mathcal{I}|\mathcal{J}}(1,1) \times \Sigma_{\mathcal{I}|\mathcal{J}}(2,2)}}.
\]
When a branch has an outage, the conditional correlation becomes zero. Fig.~\ref{fig:outage_line} shows the absolute conditional correlation $|\rho_{i,j}|$ of the loopy $8$-bus system in Fig.~\ref{fig:8bus_loop} after branch 3-4 and branch 2-6 have outages. The red boxes indicate the branches that have outages. When the post-outage distribution $f$ is known, the true $\Sigma_1$ is used to compute the conditional correlation. Comparing Fig.~\ref{fig:pre} and \ref{fig:post}, clearly, the absolute conditional corrections of outage branches change to zero after outages. The diagonal terms are the self-correlation and equal to one. This observation indicates that this proposed outage localization method is sensitive to outages and validates our proof in Theorem~\ref{thm:cond_indept}. When $f$ is unknown, by comparing Fig.~\ref{fig:pre} and \ref{fig:post-unknown}, we can still identify the outage lines. Therefore, the proposed $\mu$PMU-based method can still localize the out-of-service branches as accurate as the optimal approach.

\begin{figure} 
    \centering
  \subfloat[Pre-outage\label{fig:pre}]{%
       \includegraphics[width=0.48\linewidth]{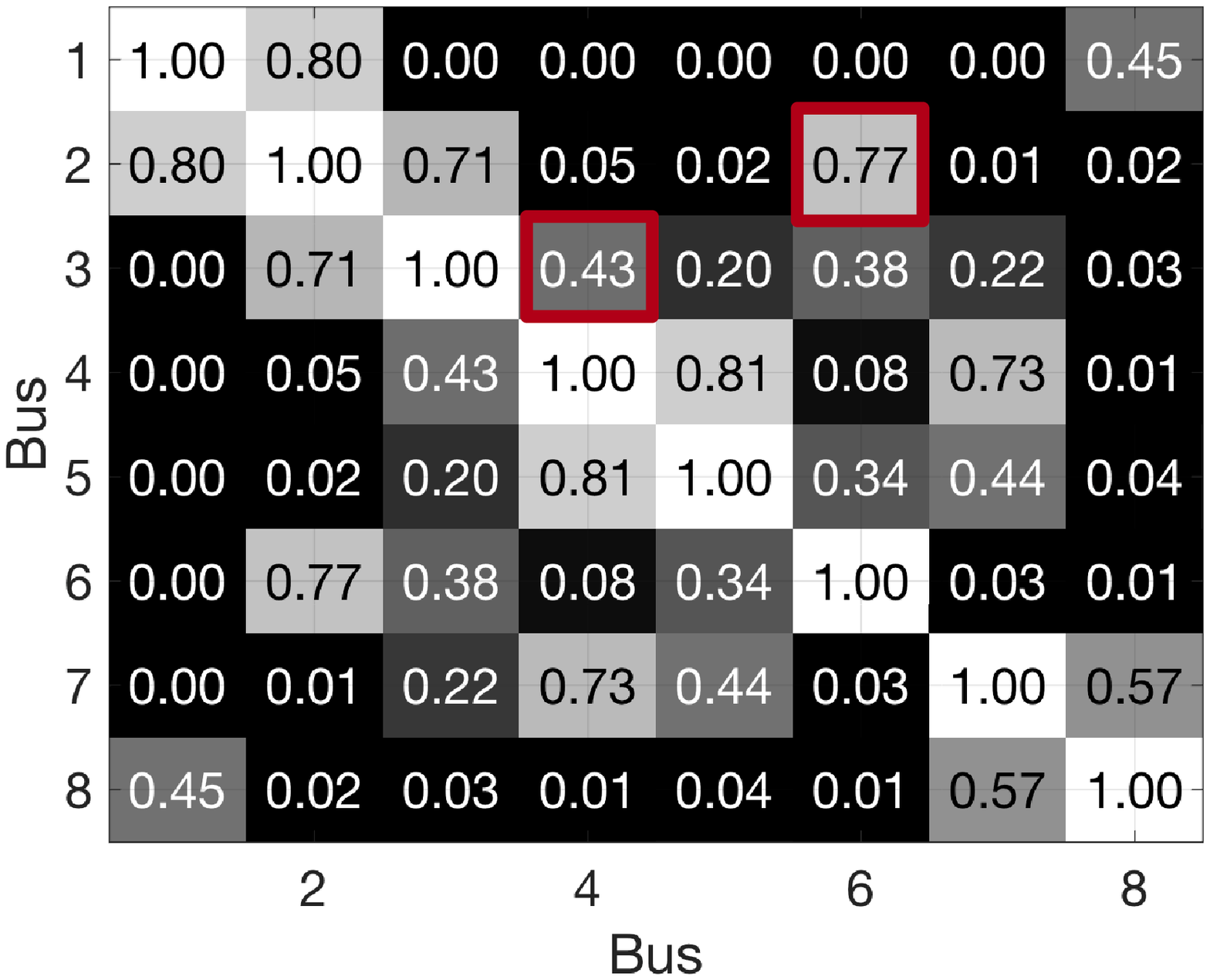}}
    \hfill
  \subfloat[Post-outage\label{fig:post}]{%
        \includegraphics[width=0.48\linewidth]{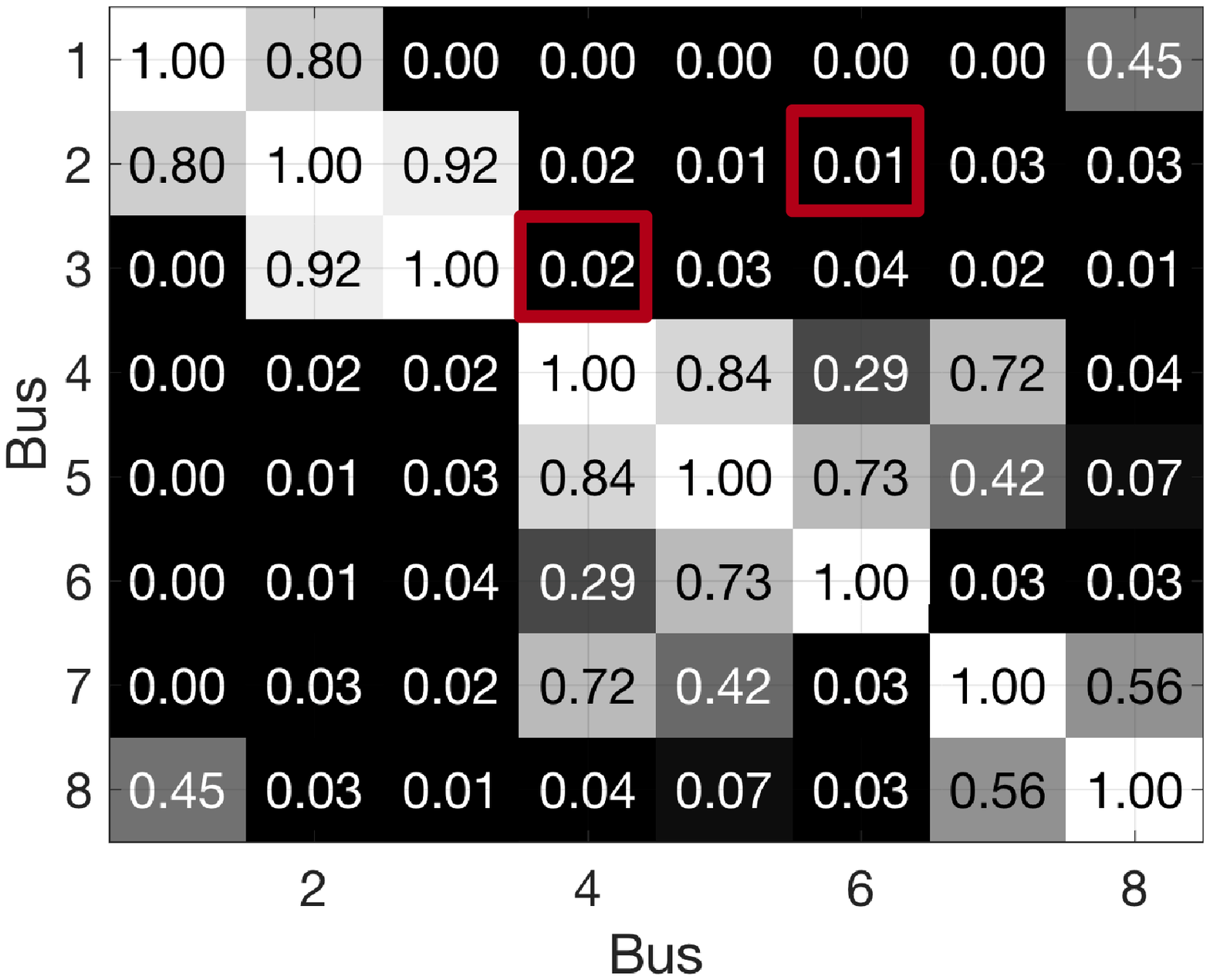}}
\\
  \subfloat[Post-outage with unknown distribution\label{fig:post-unknown}]{%
        \includegraphics[width=0.48\linewidth]{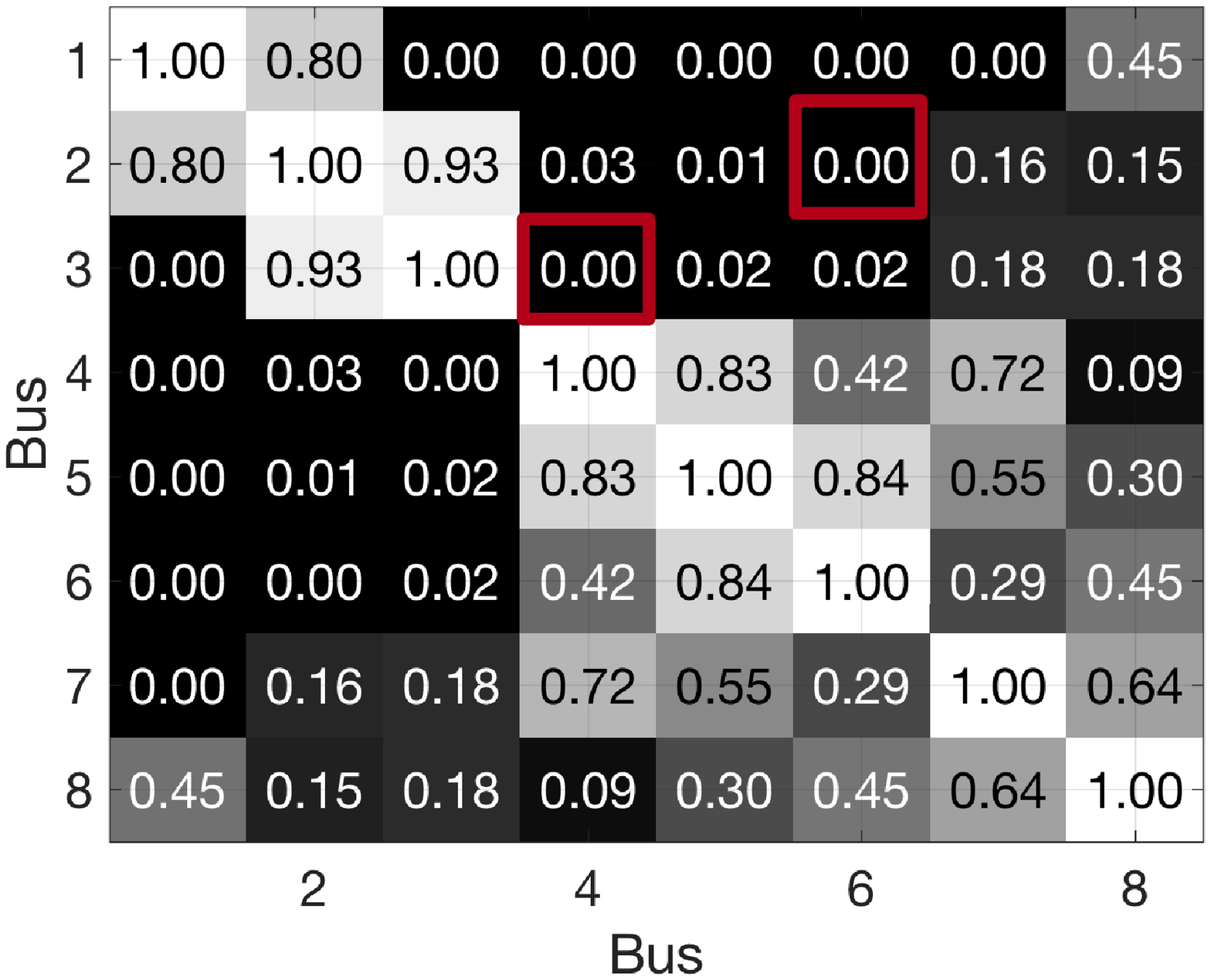}}
  \caption{Absolute conditional correlation of $8$-bus system before (a) and after (b \& c) an outage (Branches $3$-$4$ and $2$-$6$).}
    \label{fig:outage_line} 
    \vspace{-4ex}
\end{figure}

\vspace{-2ex}
\subsection{Line Outage Detection with Limited $\mu$PMU Deployment}
In previous simulation sections, we have demonstrated that the proposed algorithm can quickly identify out-of-service branches if every bus has a $\mu$PMU installed. However, as discussed in Section~\ref{sec:pmu}, requiring every bus has a $\mu$PMU is difficult in today's distribution grids. We use Corollary~\ref{thm:pmu} to prove that the optimal method in (\ref{eq:detect_rule}) can still be applied with a subset of PMUs. In this section, we will compare the performance with different number of $\mu$PMUs in the system. Fig.~\ref{fig:PMU} presents the detection delays for different number of $\mu$PMUs in the $8$-bus system. The sampling period is 1 minute. Clearly, the detection delay becomes longer with the decrease of $\mu$PMUs in the system. The reason is that with less $\mu$PMUs in the system, the KL distance becomes smaller and the detection delay in (\ref{eq:bound}) increases. However, considering the smart meters have a sampling period of 15 minutes, $\mu$PMUs still have a significant advantage when over four $\mu$PMUs are deployed in the system. For the  $123$-bus system, if we have $15$ $\mu$PMUs installed in the system, the detection delay is less than $10$ minutes. For line outage localization, we have demonstrated the alternative process in Section~\ref{sec:pmu_localize}.

\putFig{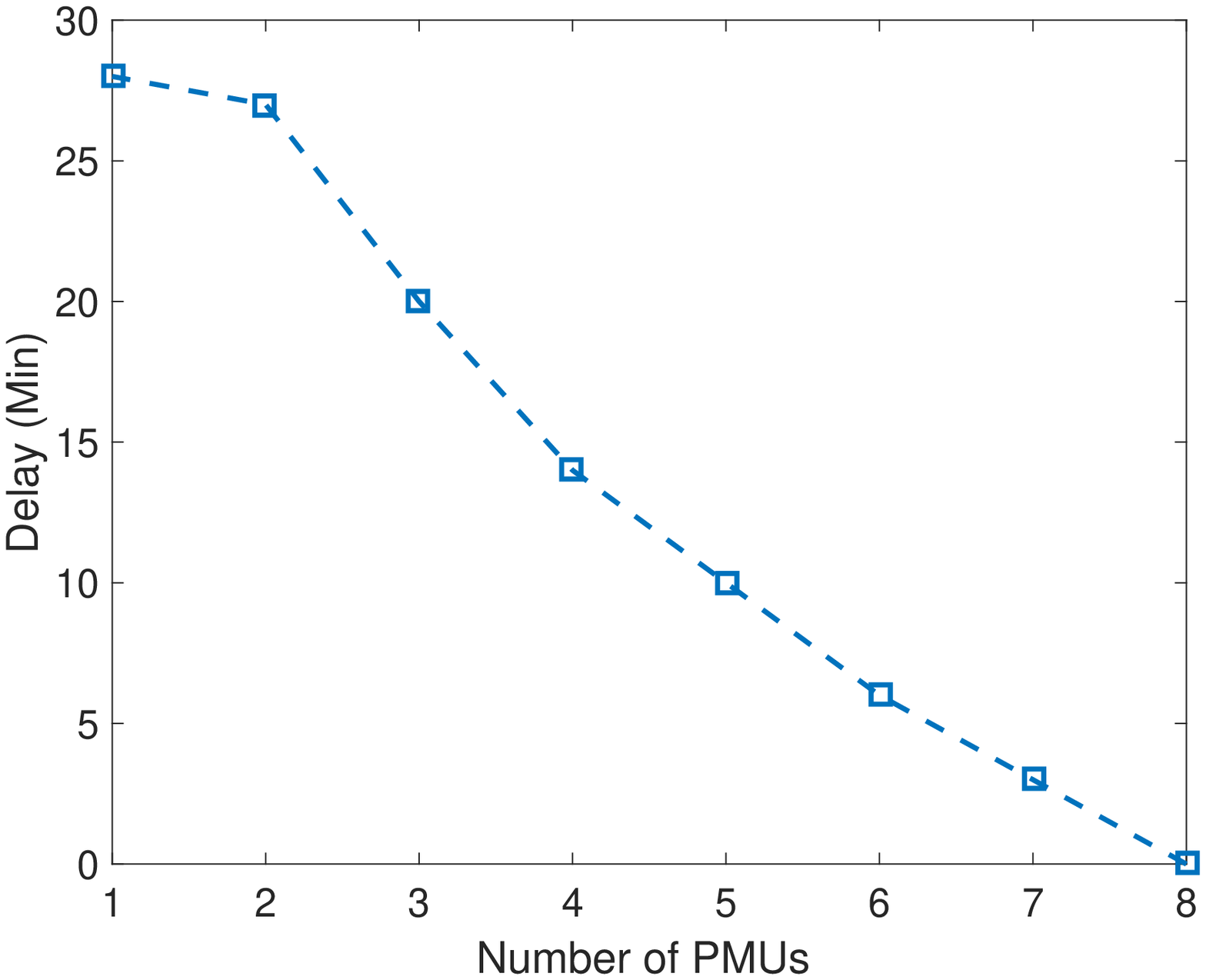}{Detection delays for different numbers of $\mu$PMUs in the $8$-bus system.}{0.75\linewidth}

\vspace{-2ex}
\section{Conclusion}\label{sec:con}
In this paper, we propose a $\mu$PMU-based data-driven algorithm to automatically detect and identify outages in distribution grids with increasing renewable penetration. Specifically, we develop a stochastic modeling of $\mu$PMU data stream and propose a change point detection approach based on the probability distribution changes because of outage events. As a highlight, unlike existing approaches, our method requires neither the grid topology nor the damage pattern as a prior. This leads to wide applicability of our proposed $\mu$PMU-based method than existing methods. In addition to outage detection, we provide theoretical prove that optimal identification can be achieved due to the conditional independence of voltages based on the power flow analysis. We verify the proposed algorithm on eight distribution grid systems with and without DERs. From extensive simulations with $\mu$PMU data, our algorithm can perfectly detect and identify outages in a short time, with and without the integration of DERs, thanks to $\mu$PMU's fast speed and high data quality. Also, with a small coverage of $\mu$PMUs, our algorithm can still quickly report outages.

\appendix
\subsection{Unknown Post-Outage Distribution Parameters Estimation}
\label{sec:parm_est}
In this section, we will prove how to estimate the unknown parameters of the distribution $f$. Since $g$ and $f$ are multi-variate Gaussian distributions, (\ref{eq:approx_log_post}) can be written as
\begin{align*}
    & P(\mathcal{H}_1|\mv^{1:N}) 
    = \log C + \sum_{k=1}^N\frac{-\pi(k)}{2} \cdot \\
    &\left(
    \sum_{n=1}^{k-1}\log|2\pi\Sigma_0| + (\mv[n] - \mu_0)^T\Sigma_0^{-1}(\mv[n] - \mu_0) \right. \\
    &+ \left. \sum_{n=k}^{N}\log|2\pi\Sigma_1| + (\mv[n] - \mu_1)^T\Sigma_1^{-1}(\mv[n] - \mu_1)
    \right).
\end{align*}
To estimate $\mu_1$, we have
\[
\frac{\partial P(\mathcal{H}_1|\mv^{1:N})}{\partial \mu_1} = \sum_{k=1}^N \frac{-\pi(k)}{2}\sum_{n=k}^N(\mv[n]-\mu_1)\Sigma_1^{-1} = 0.
\]
Since
\[
    \sum_{n=k}^N(\mv[n]-\mu_1) = \left(\sum_{n=k}^N\mv[n] - (N-k+1)\mu_1\right),
\]
the estimate of $\mu_1$ is
$
\widehat{\mu}_1 = \frac{\sum_{k=1}^N\pi(k)\sum_{n=k}^N\mv[n]}{\sum_{n=k}^N\pi(k)(N-k+1)}.
$

For the covariance matrix $\Sigma_1$, the partial derivative is 
\[
   \frac{\partial P(\mathcal{H}_1|\mv^{1:N})}{\partial \Sigma_1} = \sum_{k=1}^N\frac{-\pi(k)}{2}\left(\sum_{n=k}^N S[k] - (N-k+1)\Sigma_1\right)
\]
where $S[k] = \sum_{n=k}^N(\mv[n] - \mu_1)(\mv[n] - \mu_1)^T$. Letting $\mu_1 = \widehat{\mu}_1$ and $\partial P(\mathcal{H}_1|\mv^{1:N})/\partial \Sigma_1 = 0$, the covariance matrix estimate is 
$
\widehat{\Sigma}_1 = \frac{\sum_{k=1}^N\pi(k)S[k]}{\sum_{k=1}^N\pi(k)(N-k+1)}.
$

\bibliographystyle{IEEEtran}
\bibliography{ref}
\end{document}